\documentclass[draftclsnofoot,peerreview,english]{IEEEtran}
\usepackage[T1]{fontenc}
\usepackage[latin9]{inputenc}
\usepackage{amsmath}
\usepackage{amssymb}
\usepackage{mathrsfs}
\usepackage{bbm}
\usepackage{mathtools}
\usepackage{empheq}
\usepackage{graphicx}
\usepackage{amsthm}	
\usepackage{hyperref}
\usepackage{cite}
\usepackage{xcolor}


\usepackage{babel}

\newtheorem{thm}{Theorem}				
\newtheorem{prop}[thm]{Proposition}		
\newtheorem{lem}[thm]{Lemma}			
\newtheorem{cor}[thm]{Corollary}

\DeclareMathOperator{\pr}{\mathrm{Pr}} 		

\begin{document}

\title{The Birthday Problem and \\Zero-Error List Codes}

\author{Parham Noorzad,
Michelle Effros, 
Michael Langberg,
and Victoria Kostina%
\thanks{This paper was presented in part at the 2017 IEEE International 
Symposium of Information Theory in Aachen.}
\thanks{This material is based upon work supported by the National Science Foundation 
under Grant Numbers 1321129, 1527524, and 1526771.}%
\thanks{P. Noorzad was with the California
Institute of Technology, Pasadena, CA 91125 USA. He is now
with Qualcomm Research, San Diego, CA 92121 USA
(email: parham@qti.qualcomm.com). }%
\thanks{M. Effros and V. Kostina are with the California
Institute of Technology, Pasadena, CA 91125 USA
(emails: effros@caltech.edu, vkostina@caltech.edu). }%
\thanks{M. Langberg is with the State
University of New York at Buffalo, Buffalo, NY 14260 USA
(email: mikel@buffalo.edu).}}

\maketitle

\begin{abstract} 
As an attempt to bridge the gap between the probabilistic world of
classical information theory and the combinatorial world of
zero-error information theory,
this paper studies the performance of randomly
generated codebooks over discrete memoryless channels under a 
zero-error list-decoding constraint. This study allows the application of
tools from one area to the other.
Furthermore, it leads to an 
information-theoretic formulation of the birthday problem,
which is concerned with the probability that in a given population,
a fixed number of people have the same birthday. 
Due to the lack of a closed-form expression for this probability
when the distribution of birthdays is not uniform,
the resulting expression is not simple to analyze; in
the information-theoretic formulation, however, the
asymptotic behavior of this probability can be characterized
exactly for all distributions.

\end{abstract}

\section{Introduction} \label{sec:intro}

Finding channel capacity under a zero-error constraint
requires fundamentally different tools and ideas
from the study of capacity under an asymptotically negligible error constraint; the 
former is essentially a graph-theoretic problem \cite{Shannon56}, while
the latter mainly relies on probabilistic arguments \cite{Shannon48}.
To obtain a better understanding of the contrast between
zero-error information theory and classical information theory, 
we apply probabilistic tools to 
the study of zero-error channel coding. 

The random code construction of Shannon \cite{Shannon48} shows
that for the discrete memoryless channel 
$(\mathcal{X},W(y|x),\mathcal{Y})$, a sequence of 
rate-$R$ codebooks $(\mathscr{C}_n)_{n=1}^\infty$,
randomly generated according to distribution $P(x)$,
achieves 
\begin{equation} \label{eq:expPe}
  \lim_{n\rightarrow\infty}
  \mathbb{E}\big[P_e^{(n)}
  (\mathscr{C}_n)\big]=0,
\end{equation}
if $R<I(X;Y)$. In (\ref{eq:expPe}),
$P_e^{(n)}(\mathscr{C}_n)$ is the
average probability of error of codebook $\mathscr{C}_n$.
From Markov's inequality, it follows that $R<I(X;Y)$ 
suffices to ensure that
\begin{equation} \label{eq:PrPe}
  \forall\: \epsilon\in (0,1)\colon
  \lim_{n\rightarrow\infty}
  \pr\big\{P_e^{(n)}(\mathscr{C}_n)
  \leq\epsilon\big\}=1.
\end{equation}

Our aim is to understand the behavior of randomly generated 
codebooks when $\epsilon=0$ in (\ref{eq:PrPe}). 
Specifically, we seek to find necessary and sufficient conditions on
the rate $R$
in terms of the channel $W$ and input distribution $P(x)$, such that
the sequence of randomly generated codebooks
$(\mathscr{C}_n)_{n=1}^\infty$ satisfies
\begin{equation*}
  \lim_{n\rightarrow\infty}
  \pr\big\{P_e^{(n)}(\mathscr{C}_n)=0\big\}=1.
\end{equation*}
In other words, our goal is to quantify the performance of randomly
generated codebooks under the zero-error constraint. However, we do 
not limit ourselves to the case where the decoder only has one output.
Instead, similar to works by Elias \cite{Elias1, Elias2}, we allow 
the decoder to output a fixed number of messages.
We say a codebook corresponds to a ``zero-error $L$-list code'' if
for every message the encoder transmits, the decoder
outputs a list of size at most $L$ that contains that message.
Similar to the zero-error list capacity problem 
\cite{KornerMarton2}, this problem can be solved
using only knowledge of the distinguishability hypergraph of the channel. 
We discuss hypergraphs
and their application to zero-error list codes in Section \ref{sec:model}.
We then present our main result in Theorem \ref{thm:main} 
of Section \ref{sec:problem}, where
we provide upper and lower bounds on the rate of randomly generated 
zero-error list codes.

An important special case occurs when the channel
$W$ is an identity, that is,
\begin{equation} \label{eq:defIdChannel}
  \forall\:(x,y)\in\mathcal{X}\times\mathcal{Y}\colon
  W(y|x)=\mathbf{1}\{y=x\}.
\end{equation}
In this case, our setup leads to an information-theoretic formulation 
of the birthday problem which we next describe. 

\subsection{The Birthday Problem}

The classical birthday problem studies the probability that
a fixed number of individuals in a population have the same birthday 
under the assumption that the birthdays are independent and 
identically distributed (i.i.d.).
While this probability is simple to analyze 
when the i.i.d.\ distribution is uniform due to a closed-form expression, 
the same is not true in the non-uniform case. 
Numerical approximations for this probability
are given in \cite{GailEtAl, Nunnikhoven, Stein}.

We can frame the birthday problem as a special case of our setup above. 
Consider the problem of channel coding 
over the identity channel defined by (\ref{eq:defIdChannel}). 
Note that over this channel, a codebook corresponds to a zero-error
$L$-list code if and 
only if no group of $L+1$ messages are mapped to the same codeword.
Associating codewords with birthdays, we obtain an 
information-theoretic formulation of the birthday problem: Given a randomly generated
codebook (set of birthdays), what is the probability that some subset of 
$L+1$ codewords (birthdays) 
are identical? In Corollary \ref{cor:BDayProblem},
we provide the precise asymptotic behavior of this probability in 
terms of the R{\'e}nyi entropy of order $L+1$.

We next describe prior works that study the birthday problem
in a context similar to our work.

\subsection{Prior Works}

In \cite{Renyi2}, R{\'e}nyi 
states a result in terms of the random subdivisions of a set, 
which we reformulate in Appendix B in terms of 
zero-error $L$-list codes over the identity channel. 
R{\'e}nyi's result differs from ours in the asymptotic regime under
consideration. 

Another related work is \cite{Fujiwara}, where Fujiwara studies 
a variation of the birthday problem for the case $L=1$ in the setting of quantum
information theory. Specifically, Fujiwara determines the maximum growth rate of 
the cardinality of a sequence of codebooks that satisfy an 
``asymptotic strong orthogonality'' condition. 

Finally, we remark that the birthday problem also arises in the
context of cryptography. For a given hash function, the quantity of interest 
is the number of hash function evaluations required to find a ``collision'';
that is, two inputs that are mapped to the same output.\footnote{The attempt of finding such inputs
is referred to as the birthday attack in cryptography \cite[p. 187]{Joux}.} In this context,
the default assumption is that the hash function values are uniformly distributed as
this leads to the lowest collision probability \cite[p. 192]{Joux}. However, Bellare
and Kohno \cite{BellareKohno} argue that the uniformity assumption need not hold for
real-world hash functions. Thus, it is important to study the non-uniform case.
In \cite[Theorem 10.3]{BellareKohno}, the authors provide upper and lower bounds for
the collision probability in terms of a quantity they call ``balance,'' which is the same
as R{\'e}nyi entropy of order two up to the base of the logarithm.

In the next section, we provide an introduction into hypergraphs
and their connection to zero-error list codes. The proofs of all of
our results appear in Section \ref{sec:proofs}.

\section{Hypergraphs and Zero-Error List Codes} 
\label{sec:model}
A discrete channel is a triple
\begin{equation*}
  \big(\mathcal{X},W(y|x),\mathcal{Y}\big),
\end{equation*}
where $\mathcal{X}$ and $\mathcal{Y}$ are finite sets, 
and for each $x\in\mathcal{X}$, $W(\cdot|x)$ is a 
probability mass function on $\mathcal{Y}$. We say
an output $y\in\mathcal{Y}$ is ``reachable'' from an input
$x\in\mathcal{X}$ if $W(y|x)>0$.

A hypergraph $G=(\mathcal{V},\mathcal{E})$ consists of a
set of nodes $\mathcal{V}$ and a 
set of edges $\mathcal{E}\subseteq 2^\mathcal{V}$, where 
$2^\mathcal{V}$ denotes the collection of subsets of $\mathcal{V}$.  
We assume that $\mathcal{V}$ is finite and 
each edge has cardinality at least two. 

The distinguishability hypergraph of channel $W$, 
denoted by $G(W)$, is a hypergraph with vertex set $\mathcal{X}$ 
and an edge set $\mathcal{E}\subseteq 2^\mathcal{X}$ 
which contains collections of inputs
that are ``distinguishable'' at the decoder. 
Formally, $\mathcal{E}$ consists of all subsets 
$e\subseteq\mathcal{X}$ that satisfy
\begin{equation} \label{eq:edgeDef}
  \forall\: y\in\mathcal{Y}\colon
  \prod_{x\in e} W(y|x)=0;
\end{equation}
that is, $e\subseteq\mathcal{X}$ is an edge if no $y\in\mathcal{Y}$
is reachable from all $x\in e$. 
Note that $G(W)$ has the property that the superset of any edge is an
edge; that is, if $e\in\mathcal{E}$ and
$e\subseteq e'\subseteq\mathcal{X}$, 
then $e'\in\mathcal{E}$. Proposition \ref{prop:distinguishabilityGraphsCharacterization},
below, shows that any hypergraph $G$ with this property is the distinguishability
hypergraph of some channel $W$.


An independent set of a 
hypergraph $G=(\mathcal{V},\mathcal{E})$ is a 
subset $\mathcal{I}\subseteq\mathcal{V}$
such that no subset of $\mathcal{I}$ is in $\mathcal{E}$.
For the channel $(\mathcal{X},W(y|x),\mathcal{Y})$,
an independent set of $G(W)$ corresponds to 
a collection of inputs $\mathcal{I}\subseteq\mathcal{X}$ 
for which there exists an output $y\in\mathcal{Y}$
that is reachable from any $x\in\mathcal{I}$.  
The hypergraph $G$ is complete multipartite if 
there exists a partition $\{\mathcal{I}_j\}_{j=1}^k$ of $\mathcal{V}$
such that each $\mathcal{I}_j$ is an independent set, and
for every subset $e\subseteq \mathcal{V}$, 
either $e\in\mathcal{E}$, or $e\subseteq\mathcal{I}_j$ for some
$1\leq j\leq k$. 

As an example, consider a deterministic channel $(\mathcal{X},W(y|x),\mathcal{Y})$,
where for some mapping $\varphi\colon \mathcal{X}\rightarrow\mathcal{Y}$,
\begin{equation*}
  W(y|x)=\mathbf{1}\{y=\varphi(x)\}.
\end{equation*}
For this channel, $G(W)$ is a complete multipartite hypergraph.
Specifically, the sets $\{\varphi^{-1}(y)\}_{y\in\mathcal{Y}}$ are the independent
components of $G$, where for $y\in\mathcal{Y}$,
\begin{equation*}
  \varphi^{-1}(y):=\big\{
  x\in\mathcal{X}\big|\varphi(x)=y\big\}.
\end{equation*}
The next proposition gives a complete characterization of hypergraphs
$G$ that correspond to the distinguishability hypergraphs of
arbitrary and deterministic channels, respectively.

\begin{prop} \label{prop:distinguishabilityGraphsCharacterization}
Consider a hypergraph $G=(\mathcal{V},\mathcal{E})$. Then there exists
a discrete channel $(\mathcal{X},W(y|x),\mathcal{Y})$ such that 
$G=G(W)$ if and only if the superset of every edge of $G$ is an edge. 
Furthermore, there exists a deterministic channel $W$ such that $G=G(W)$
if and only if $G$ is complete multipartite. 
\end{prop}

Given the connection between channels and hypergraphs in 
Proposition \ref{prop:distinguishabilityGraphsCharacterization},
we now find a graph-theoretic condition for a
mapping to be a zero-error list code. We present this condition
in Proposition \ref{prop:existence}. Prior to that, we define
necessary notation.

For positive
integers $i$ and $j$ with $j\geq i$, 
$[i:j]$ denotes the set $\{i,\dots,j\}$. When $i=1$, we 
denote $[1:j]$ by
$[j]$. For example, $[1]=\{1\}$ and 
$[2:4]=\{2,3,4\}$. For any set $\mathcal{A}$ and 
nonnegative integer $k\leq |\mathcal{A}|$,
define the set
\begin{equation*}
  \binom{\mathcal{A}}{k}:=
  \big\{\mathcal{B}\big|\mathcal{B}\subseteq\mathcal{A},
  |\mathcal{B}|=k\big\}.
\end{equation*}

An $(M,L)$ list code for the channel $(\mathcal{X},W(y|x),\mathcal{Y})$ 
consists of an encoder 
\begin{equation*}
  f\colon [M]\rightarrow \mathcal{X},
\end{equation*} 
and a decoder 
\begin{equation*}
  g\colon \mathcal{Y}\rightarrow\bigcup_{\ell=1}^L\binom{[M]}{\ell}.
\end{equation*}
The pair $(f,g)$ is an
$(M,L)$ \emph{zero-error} list code for channel $W$ if for 
every $m\in [M]$ and $y\in\mathcal{Y}$
satisfying $W(y|f(m))>0$, we have $m\in g(y)$. 

Proposition \ref{prop:existence} provides a
necessary and sufficient condition for the existence
of an $(M,L)$ zero-error list code for $W$ in terms of
its distinguishability hypergraph $G(W)$.

\begin{prop} \label{prop:existence}
Consider a discrete channel 
$(\mathcal{X},W(y|x),\mathcal{Y})$, 
positive integers $M$ and $L$, and 
an encoder
\begin{equation*}
  f\colon [M]\rightarrow \mathcal{X}.
\end{equation*} 
For this encoder, a decoder 
\begin{equation*}
  g\colon \mathcal{Y}\rightarrow\bigcup_{\ell=1}^L\binom{[M]}{\ell}
\end{equation*}
exists such that the pair $(f,g)$ is an
$(M,L)$ zero-error list code for $W$
if and only if the image of every $(L+1)$-subset 
$\{m_\ell\}_{\ell=1}^{L+1}$ of $[M]$ under $f$ 
is an edge of $G(W)$. 
\end{prop}

Proposition \ref{prop:existence} reduces the existence of an
$(M,L)$ zero-error list code to the existence of an encoder with a certain
property. Because of this, henceforth we say a mapping $f\colon [M]\rightarrow\mathcal{X}$ 
is an $(M,L)$ zero-error list code for the channel $W$ if it satisfies 
the condition stated in Proposition \ref{prop:existence}. 

For each positive integer $n$, the $n$th
extension channel of $W$ is the channel 
\begin{equation*}
  \big(\mathcal{X}^n,W^n(y^n|x^n),\mathcal{Y}^n\big),
\end{equation*}
where 
\begin{equation*}
  W^n(y^n|x^n):=\prod_{t\in[n]} W(y_t|x_t). 
\end{equation*}
An $(M,L)$ zero-error list code for $W^n$ is 
referred to as an $(M,n,L)$ zero-error list code for $W$.
It is possible to show that 
the distinguishability hypergraph of $W^n$, $G(W^n)$, equals
$G^n(W)$, the $n$th co-normal power of $G(W)$ \cite{Simonyi}.
For any positive integer $n$ and any hypergraph 
$G=(\mathcal{V},\mathcal{E})$, the $n$th co-normal power of $G$,
which we denote by $G^n$, 
is defined on the set of nodes $\mathcal{V}^n$ as follows. 
For each $k\geq 2$,
the $k$-subset $e=\{v_1^n,\dots,v_k^n\}\subseteq \mathcal{V}^n$ 
is an edge of $G^n$ if for at least one $t\in [n]$, 
$\{v_{1t},\dots,v_{kt}\}\in\mathcal{E}$. This definition is motivated by
the fact that $k$ codewords are distinguishable if 
and only if their components are distinguishable in at least
one dimension. 

\section{Random Zero-Error List Codes} \label{sec:problem}

Fix a sequence of probability mass functions $(P_n(x^n))_{n=1}^\infty$,
where $P_n$ is defined over $\mathcal{X}^n$.  
Our aim here is to study the 
performance of the sequence
of random codes
$F_n\colon [M_n]\rightarrow\mathcal{X}^n$ over
the channel $(\mathcal{X},W(y|x),\mathcal{Y})$, where 
\begin{equation*}
  F_n(1),\dots,F_n(M_n)
\end{equation*} 
are $M_n$ i.i.d.\ random variables, and
\begin{equation*}
  \forall\: m\in [M_n]\colon\pr\big\{F_n(m)=x^n\big\}:=P_n(x^n). 
\end{equation*}
We seek to find conditions on the sequence 
$(M_n)_{n=1}^\infty$ such that  
\begin{equation*} 
  \lim_{n\rightarrow\infty}
  \pr\big\{F_n\text{ is an }(M_n,n,L)\text{ zero-error list code for }W
  \big\}=1.
\end{equation*}
Theorem \ref{thm:main}, below, provides the desired conditions. 
The conditions
rely on a collection of functions of the pair $(G^n(W),P_n)$, denoted
by 
\begin{equation*}
  (\theta_{L+1}^{(\ell)}(G^n(W),P_n))_{\ell=1}^{L+1},
\end{equation*}
which we next define.
 
Consider a hypergraph $G=(\mathcal{V},\mathcal{E})$. 
For any positive integer $k$, let $v_{[k]}=(v_1,\dots,v_k)$
denote an element of $\mathcal{V}^k$. For all $v_{[k]}\in\mathcal{V}^k$
and every nonempty subset $S\subseteq [k]$, define $v_S:=(v_j)_{j\in S}$.
Let $P$ be a probability mass function on $\mathcal{V}$ and set
\begin{equation*}
  P(v_S):=\prod_{j\in S} P(v_j). 
\end{equation*}
In addition, for each positive integer $k\geq 2$,
define the mapping $\sigma_k\colon\mathcal{V}^k\rightarrow 2^\mathcal{V}$
as 
\begin{equation*}
  \sigma_k(v_{[k]}):=\{v_1,\dots,v_k\}.
\end{equation*}
In words, $\sigma_k$ maps each vector 
$v_{[k]}\in\mathcal{V}^k$ to
the set containing its distinct components. 
For example, if $v_{[k]}=(v,\dots,v)$ for some $v\in\mathcal{V}$, 
then $\sigma_k(v_{[k]})=\{v\}$. 
When the value of $k$ is clear from context, 
we denote $\sigma_k$ with $\sigma$.

We next define functions of the pair $(G,P)$ that are instrumental
in characterizing the performance of random codebooks over 
channels with zero-error constraints.
For every positive integer $L$, define the quantity 
$I_{L+1}(G,P)$ as 
\begin{equation} \label{eq:defIkGP}
  I_{L+1}(G,P):=-\frac{1}{L}
  \log\sum_{v_{[L+1]}:\sigma(v_{[L+1]})\notin \mathcal{E}}P(v_{[L+1]}),
\end{equation}
where $\log$ is the binary logarithm.
Note that in (\ref{eq:defIkGP}),
\begin{equation*}
  \sum_{v_{[L+1]}:\sigma(v_{[L+1]})\notin \mathcal{E}}P(v_{[L+1]})
\end{equation*}
equals the probability of independently selecting $L+1$ vertices of $G$,
with replacement, that are indistinguishable. The negative
sign in (\ref{eq:defIkGP}) results in the nonnegativity 
of $I_{L+1}(G,P)$; division by $L$, as we show in 
Proposition \ref{prop:thetaProperties}, makes it comparable
to the R{\'e}nyi entropy of order $L+1$ \cite{Renyi},
which is defined as
\begin{equation*}
  H_{L+1}(P):=-\frac{1}{L}\log\sum_{v\in\mathcal{V}}
  \big(P(v)\big)^{L+1}.
\end{equation*}

We now define the sequence of functions
$(\theta^{(\ell)}_{L+1}(G,P))_{\ell\in [L+1]}$. This sequence
arises from the application of a second moment
bound in the proof of Theorem \ref{thm:main}
given in Subsection \ref{subsec:thmMain}. 
Set $\theta_{L+1}^{(L+1)}(G,P):=I_{L+1}(G,P)$, and
for $\ell\in [L]$, let  
\begin{equation} \label{eq:defTkGP}
  \theta_{L+1}^{(\ell)}(G,P) 
  :=2I_{L+1}(G,P) 
  +\frac{1}{L}\log
  \sum_{v_{[\ell]}}P(v_{[\ell]})\Bigg[
  \sum_{v_{[\ell+1:L+1]}:\sigma(v_{[L+1]})\notin 
  \mathcal{E}}P(v_{[\ell+1:L+1]})\Bigg]^2.
\end{equation}
The following proposition describes a number of properties that the
sequence $(\theta^{(\ell)}_{L+1}(G,P))_{\ell\in [L+1]}$ satisfies.

\begin{prop} \label{prop:thetaProperties}
For every hypergraph $G=(\mathcal{V},\mathcal{E}),$
probability mass function
$P$ on $\mathcal{V},$ and positive integer $L,$
the following statements hold.

(i) For all $\ell\in [L+1],$ 
\begin{equation*}
  0\leq\theta^{(\ell)}_{L+1}(G,P)\leq I_{L+1}(G,P).
\end{equation*}

(ii) We have 
\begin{equation*}
  0\leq I_{L+1}(G,P)\leq H_{L+1}(P).
\end{equation*}
Let $\mathrm{supp}(P)$ denote the support of $P$. 
Then
\begin{align*}
  I_{L+1}(G,P)=0&\iff\forall\: e\subseteq\mathrm{supp}(P)\colon
  \big(2\leq |e|\leq {L+1}\implies e\notin\mathcal{E}\big)\\
  I_{L+1}(G,P)=H_{L+1}(P)&\iff\forall\: e\subseteq\mathrm{supp}(P)\colon
  \big(2\leq |e|\leq {L+1}\implies e\in\mathcal{E}\big).
\end{align*}

(iii) For every positive integer $n\geq 2,$
define the probability mass function $P^n$ 
on $\mathcal{V}^n$ as
\begin{equation*}
  \forall\: v^n\in\mathcal{V}^n\colon
  P^n(v^n) := \prod_{t\in[n]} P(v_t).
\end{equation*}
Then for all $\ell\in [L+1],$
\begin{equation*}
  \theta^{(\ell)}_{L+1}(G^n,P^n)=
  n\theta^{(\ell)}_{L+1}(G,P),
\end{equation*}
where $G^n$ is the $n$th co-normal power of $G$
defined in Section \ref{sec:model}. 
\end{prop}
For $L=1,$ $I_{L+1}(G,P)$ has further properties which
we discuss in Appendix A.

We next state our main result
which provides upper and lower bounds on the cardinality
of a randomly generated codebook that has zero error. 

\begin{thm} \label{thm:main}
Consider a channel $(\mathcal{X},W(y|x),\mathcal{Y})$ 
and a sequence of probability mass functions $(P_n(x^n))_{n=1}^\infty$.
If 
\begin{equation} \label{eq:rateIPG}
  \lim_{n\rightarrow\infty} M_n^{L+1}2^{-LI_{L+1}(G^n(W),P_n)}=0,
\end{equation}
then 
\begin{equation} \label{eq:limit}
  \lim_{n\rightarrow\infty}
  \pr\big\{F_n\text{ is an }(M_n,n,L)\text{ zero-error list code}
  \big\}=1.
\end{equation}
Conversely, assuming (\ref{eq:limit}), then
for some $\ell\in [L+1]$, 
\begin{equation} \label{eq:thetaThm}
  \lim_{n\rightarrow\infty}M_n^\ell 
  2^{-L\theta_{L+1}^{(\ell)}(G^n(W),P_n)}=0.
\end{equation}
\end{thm}

In Theorem \ref{thm:main}, if a channel $W$ 
and a sequence of probability mass functions 
$(P_n)_{n=1}^\infty$ satisfy
\begin{equation} \label{eq:equality}
  \max_{\ell\in [L+1]}\frac{1}{\ell}
  \theta_{L+1}^{(\ell)}(G^n(W),P_n)=
  \frac{1}{L+1}I_{L+1}(G^n(W),P_n),
\end{equation}
for sufficiently large $n$, 
then (\ref{eq:rateIPG}), in addition to
being sufficient for (\ref{eq:limit}), is necessary as
well. In the next corollary, we present a sufficient condition
under which (\ref{eq:equality}) holds. To describe this condition
precisely, we require the next definition.

%
%
%

Consider a hypergraph $G=(\mathcal{V},\mathcal{E})$ and a probability mass function
$P$ on $\mathcal{V}$. Let $\mathcal{V}_P=\mathrm{supp}(P)$
and $\mathcal{E}_P\subseteq\mathcal{E}$ be the set of all
edges whose vertices lie in $\mathcal{V}_P$. We then refer to 
the hypergraph $G_P\coloneqq (\mathcal{V}_P,\mathcal{E}_P)$ as the 
subhypergraph of $G$ induced by $P$. 
For a fixed $n$, a sufficient condition for (\ref{eq:equality}) 
to hold is for the subhypergraph of $G^n(W)$ induced 
by $P_n$ be complete multipartite. (Recall definition from Section
\ref{sec:model}.) This results in the 
next corollary. The proof of this corollary,
together with the sufficient condition for (\ref{eq:equality}), appears in
Subsection \ref{subsec:equality}. 

\begin{cor} \label{cor:equality}
Consider a channel $(\mathcal{X},W(y|x),\mathcal{Y})$
and a sequence of probability mass functions $(P_n(x^n))_{n=1}^\infty$.  
If for sufficiently large $n$, the subhypergraph of $G^n(W)$
induced by $P_n$ is complete multipartite, then 
\begin{align*}
  \MoveEqLeft
  \lim_{n\rightarrow\infty}
  \pr\big\{F_n\text{ is an }(M_n,n,L)\text{ zero-error list code}
  \big\}=1\\
  &\iff 
  \lim_{n\rightarrow\infty} M_n^{L+1}2^{-LI_{L+1}(G^n(W),P_n)}=0.
\end{align*}
\end{cor}

One scenario where the sufficient condition of Corollary
\ref{cor:equality} holds automatically for all $n\geq 1$ is when
$G(W)$ is complete multipartite. This is stated in the next 
lemma.

\begin{lem} \label{lem:completeMultipartite}
If $G$ is a complete multipartite hypergraph, then
for all $n\geq 2$, so is $G^n$. 
\end{lem} 

In the case where $G(W)$ is not complete multipartite, in order
to obtain a simpler version of Theorem \ref{thm:main}, we  
assume that the codebook distribution is not only i.i.d.\
across messages, but also over time. In addition, we assume that
the message set cardinality grows exponentially in the blocklength.
Formally, we fix a probability mass function $P$ on $\mathcal{X}$ and
a rate $R\geq 0$. Then, in Theorem \ref{thm:main}, by setting $P_n:=P^n$ 
and $M_n:=\lfloor 2^{nR}\rfloor$ for all positive integers
$n$, and applying Parts (i) and (iii) of Proposition
\ref{prop:thetaProperties}, we get the following corollary.
\begin{cor} \label{cor:IIDcase}
Consider a channel $(\mathcal{X},W(y|x),\mathcal{Y})$ 
and a probability mass function $P$ on $\mathcal{X}$.
If 
\begin{equation*}  
  R<\frac{L}{L+1}I_{L+1}(G,P),
\end{equation*}
then 
\begin{equation} \label{eq:limit2}
  \lim_{n\rightarrow\infty}
  \pr\big\{F_n\text{ is an }(2^{nR},n,L)\text{ zero-error list code for }
  W \big\}=1.
\end{equation}
Conversely, if (\ref{eq:limit2}) holds, then
\begin{equation} \label{eq:upperBoundLIL}
  R< LI_{L+1}(G,P).
\end{equation}
\end{cor}

Note that in Corollary \ref{cor:IIDcase},
if $G(W)$ is complete multipartite, as in 
the next example, then using Corollary \ref{cor:equality},
the upper bound (\ref{eq:upperBoundLIL})
can be improved to 
\begin{equation*}  
  R<\frac{L}{L+1}I_{L+1}(G,P).
\end{equation*}

We next apply Corollary \ref{cor:equality} to the 
identity channel $W=(\mathcal{X},\mathbf{1}\{y=x\},\mathcal{X})$.
Per the informal discussion in the Introduction, applying our
result to this channel gives the exact asymptotic behavior of
the probability of coinciding birthdays in the birthday problem. 
We now formalize this connection.

Note that every subset $e$
of $\mathcal{X}$ with $|e|\geq 2$ is an edge of $G(W)$. 
Thus for $n\geq 2$, every $e\subseteq\mathcal{X}^n$
with $|e|\geq 2$ is an edge of $G^n(W)$. 
Therefore, for distinct messages $m_1,\dots, m_{L+1}\in [M_n]$, 
we have $F_n(m_1)=\dots=F_n(m_{L+1})$
if and only if
\begin{equation*}
  \big(F_n(m_1),\dots,F_n(m_{L+1})\big)
  \text{ is not an edge in }G^n(W).
\end{equation*}
Hence Proposition \ref{prop:existence} implies that 
(\ref{eq:injective}) holds if and only if 
\begin{equation} \label{eq:bDayLimit}
  \lim_{n\rightarrow\infty}
  \pr\big\{F_n\text{ is an }(M_n,n,L)\text{ zero-error list code}
  \big\}=1.
\end{equation}
Now from Corollary \ref{cor:equality}, it follows that (\ref{eq:bDayLimit})
holds if and only if 
\begin{equation} \label{eq:OurResultBDay}
  \lim_{n\rightarrow\infty} M_n^{L+1} 2^{-LH_{L+1}(P_n)}=0.
\end{equation} 
This proves the next corollary.

\begin{cor} \label{cor:BDayProblem}
Fix an integer $L\geq 1$, a finite set $\mathcal{X}$, and a sequence of 
probability mass functions $(P_n(x^n))_{n=1}^\infty$. For each $n$, let
$F_n\colon [M_n]\rightarrow\mathcal{X}^n$ be a random mapping 
with i.i.d.\ values and distribution $P_n(x^n)$; that is, 
\begin{equation*}
  \forall\: m\in [M_n]\colon
  \pr\big\{F_n(m)=x^n\big\}=P_n(x^n).
\end{equation*}
Then we have 
\begin{equation} \label{eq:injective}
  \lim_{n\rightarrow\infty}
  \pr\Big\{\exists\: m_1,\dots,m_{L+1}\in [M_n]\colon
  F_n(m_1)=\dots=F_n(m_{L+1})\Big\}=0
\end{equation}
if and only if
\begin{equation*}
  \lim_{n\rightarrow\infty} M_n^{L+1} 2^{-LH_{L+1}(P_n)}=0.
\end{equation*} 
\end{cor}

In words, to guarantee the absence of collisions of $(L+1)$-th order, 
the population size $M_n$ must be negligible compared to 
$2^{\frac{L}{L+1}H_{L+1}(P_n)}$. 

\section{Proofs} \label{sec:proofs}

In this section, we provide detailed proofs of our results.

\subsection{Proof of Proposition \ref{prop:distinguishabilityGraphsCharacterization}}
For each of the two cases, one direction is proved before the statement of
the proposition in Section \ref{sec:model}. Here we prove the reverse direction
of each case.

Suppose $G=(\mathcal{V},\mathcal{E})$ is a hypergraph where the superset of every edge
is an edge. We define a channel 
\begin{equation*}
  \big(\mathcal{X},W(y|x),\mathcal{Y}\big)
\end{equation*}
such that $G=G(W)$. Set
\begin{align*}
  \mathcal{X} &\coloneqq  \mathcal{V}\\
  \mathcal{Y} &\coloneqq  2^\mathcal{V}\setminus\mathcal{E}.
\end{align*}
Note that $\mathcal{Y}$ is not empty, since by definition, each 
edge has cardinality at least two. Define $W$ as
\begin{equation*}
  W(y|x)\coloneqq 
  \frac{\mathbf{1}\{x\in y\}}{|\{\bar y\in\mathcal{Y}:x\in \bar y\}|}.
\end{equation*}
Then for every subset $e\subseteq\mathcal{X}$ and every $y\in\mathcal{Y}$,
\begin{equation} \label{eq:productNotZero}
  \prod_{x\in e} W(y|x)
  =\prod_{x\in e}
  \frac{\mathbf{1}\{x\in y\}}{|\{\bar y\in\mathcal{Y}:x\in \bar y\}|}
  \neq 0
\end{equation}
if and only if $e\subseteq y$. Since by definition of $\mathcal{Y}$,
$y$ is not an edge, and by assumption, the superset of every edge is an edge,
(\ref{eq:productNotZero}) holds for some $y\in\mathcal{Y}$ if and only if 
$e\notin\mathcal{E}$. Thus $G=G(W)$.

Next assume $G=(\mathcal{V},\mathcal{E})$ is complete multipartite;
that is, there exists a partition $\{\mathcal{I}_j\}_{j=1}^k$ of $\mathcal{V}$
such that each $\mathcal{I}_j$ is an independent set, and
for every subset $e\subseteq \mathcal{V}$, 
either $e\in\mathcal{E}$, or $e\subseteq\mathcal{I}_j$ for some
$1\leq j\leq k$. For this hypergraph, we define a
\emph{deterministic} channel
\begin{equation*}
  \big(\mathcal{X},W(y|x),\mathcal{Y}\big)
\end{equation*}
such that $G=G(W)$. Set
\begin{align*}
  \mathcal{X} &\coloneqq  \mathcal{V}
  =\bigcup_{j\in [k]}\mathcal{I}_j\\
  \mathcal{Y} &\coloneqq [k],
\end{align*}
and define $W$ as
\begin{equation*}
  W(y|x)\coloneqq \mathbf{1}\{x\in \mathcal{I}_y\}.
\end{equation*}
Then for every subset $e\subseteq\mathcal{X}$ and every $y\in\mathcal{Y}$,
\begin{equation} \label{eq:productNotZeroDetCh}
  \prod_{x\in e} W(y|x)
  =\prod_{x\in e} \mathbf{1}\{x\in \mathcal{I}_y\}
  \neq 0
\end{equation}
if and only if $e\subseteq \mathcal{I}_y$. By assumption, however,
every $e\in\mathcal{E}$ is either in $\mathcal{E}$ or is a subset of an
independent set $\mathcal{I}_y$ for some $y\in\mathcal{Y}$. Thus
(\ref{eq:productNotZeroDetCh}) holds for some $y\in\mathcal{Y}$ if and
only if $e\notin\mathcal{E}$. This completes the proof.

\subsection{Proof of Proposition \ref{prop:existence}} 
\label{subsec:existence}
Let $(f,g)$ be an $(M,L)$ zero-error list
code for channel $W$. If $[M]$ has a subset of cardinality $L+1$, 
say $\{m_\ell\}_{\ell=1}^{L+1}$, such that
$\{f(m_\ell)\}_{\ell=1}^{L+1}$ is not an edge in $G(W)$,
then for some $y\in\mathcal{Y}$,
\begin{equation*}
  \prod_{\ell\in [L+1]}
  W(y|f(m_\ell))>0.
\end{equation*}
Thus for every $\ell\in [L+1]$, 
$W(y|f(m_\ell))>0$, which implies $m_\ell\in g(y)$. Thus
$g(y)$ contains at least $L+1$ distinct elements, which
is a contradiction. 

Conversely, suppose we have an encoder 
$f\colon [M]\rightarrow\mathcal{X}$
that maps every $(L+1)$-subset
of $[M]$ onto an edge of $G(W)$. 
For each $y\in\mathcal{Y}$, define the set 
\begin{equation*}
  \mathcal{M}_y\coloneqq
  \Big\{m\in [M]\Big|W(y|f(m))>0\Big\}.
\end{equation*}
Suppose for some $y^*\in\mathcal{Y}$,
$|\mathcal{M}_{y^*}|>L$. Then $\mathcal{M}_{y^*}$
has a subset $\mathcal{A}$ of cardinality $L+1$. By
assumption, $f$ maps $\mathcal{A}$ to an edge of 
$G(W)$, which implies that for all $y\in\mathcal{Y}$,
including $y=y^*$, 
\begin{equation*}
  \prod_{m\in\mathcal{A}}W\big(y|f(m)\big)=0.
\end{equation*} 
This contradicts the definition of $M_{y^*}$. Thus
for all $y\in\mathcal{Y}$, $|\mathcal{M}_y|\leq L$.

Now if we define the decoder as
\begin{equation*}
  \forall\: y\in\mathcal{Y}\colon
  g(y)=\mathcal{M}_y,
\end{equation*}
then the pair $(f,g)$ is an $(M,L)$ zero-error list code
and the proof is complete. 

\subsection{Proof of Proposition \ref{prop:thetaProperties}}
\label{subsec:thetaProperties}

(i) We prove the nonnegativity of $\theta^{(\ell)}_{L+1}(G,P)$
first for $\ell=L+1$ and then for arbitrary $\ell\in [L]$. 
Recall that $\theta_{L+1}^{(L+1)}(G,P)=I_{L+1}(G,P)$. We have
\begin{equation*}
  \sum_{v_{[L+1]}:\sigma(v_{[L+1]})\notin\mathcal{E}}P(v_{[L+1]})
  \leq \sum_{v_{[L+1]}\in\mathcal{V}^{L+1}}P(v_{[L+1]})=
  \Big(\sum_{v\in\mathcal{V}} P(v)\Big)^{L+1}=1,
\end{equation*}
which implies
\begin{equation*}
  \theta_{L+1}^{(L+1)}(G,P)=I_{L+1}(G,P)
  =-\frac{1}{L}\log
  \sum_{v_{[L+1]}:\sigma(v_{[L+1]})\notin\mathcal{E}}
  P(v_{[L+1]})\geq 0. 
\end{equation*}

For $\ell\in [L]$, rewrite $\theta^{(\ell)}_{L+1}(G,P)$ as 
\begin{equation*}
  \theta^{(\ell)}_{L+1}(G,P)
  =\frac{1}{L}\log
  \frac{\sum_{v_{[\ell]}}P(v_{[\ell]})
  \Big[\sum_{v_{[\ell+1:L+1]}:\sigma(v_{[L+1]})\notin\mathcal{E}}
  P(v_{[\ell+1:L+1]})\Big]^2}{\Big[\sum_{v_{[L+1]}:\sigma(v_{[L+1]})\notin
  \mathcal{E}}P(v_{[L+1]})\Big]^2}.
\end{equation*}
Note that
\begin{equation*}
  \sum_{v_{[\ell]}\in\mathcal{V}^\ell}P(v_{[\ell]})
  =\Big(\sum_{v\in\mathcal{V}} P(v)\Big)^\ell=1.
\end{equation*}
Therefore, by the Cauchy-Schwarz inequality,
\begin{align*}
  \sum_{v_{[\ell]}}P(v_{[\ell]})
  \bigg[\sum_{v_{[\ell+1:L+1]}:\sigma(v_{[L+1]})\notin\mathcal{E}}
  P(v_{[\ell+1:L+1]})\bigg]^2
  &\geq \bigg[\sum_{v_{[\ell]}}P(v_{[\ell]})
  \sum_{v_{[\ell+1:L+1]}:\sigma(v_{[L+1]})\notin\mathcal{E}}
  P(v_{[\ell+1:L+1]})\bigg]^2\\ &\geq
  \bigg[\sum_{v_{[L+1]}:\sigma(v_{[L+1]})\notin
  \mathcal{E}}P(v_{[L+1]})\bigg]^2,
\end{align*}
which implies $\theta^{(\ell)}_{L+1}(G,P)\geq 0$.

We next prove the upper bound on $\theta^{(\ell)}_{L+1}(G,P)$. 
Note that 
\begin{align*}
  \sum_{v_{[\ell]}}P(v_{[\ell]})
  \bigg[\sum_{v_{[\ell+1:L+1]}:\sigma(v_{[L+1]})\notin\mathcal{E}}
  P(v_{[\ell+1:L+1]})\bigg]^2
  &\leq 
  \sum_{v_{[\ell]}}P(v_{[\ell]})
  \bigg[\sum_{v_{[\ell+1:L+1]}:\sigma(v_{[L+1]})\notin\mathcal{E}}
  P(v_{[\ell+1:L+1]})\bigg]\\
  &=2^{-LI_{L+1}(G,P)}.
\end{align*}
Thus
\begin{equation*}
  \theta^{(\ell)}_{L+1}(G,P)\leq 2I_{L+1}(G,P)-I_{L+1}(G,P)=I_{L+1}(G,P).
\end{equation*}

(ii) The inequality $I_k(G,P)\geq 0$ is proved in (i). Equality
holds if and only if 
\begin{equation*}
  \forall\: v_{[L+1]}\in\big(\mathrm{supp}(P)\big)^{L+1}\colon
  \sigma(v_{[L+1]})\notin\mathcal{E},
\end{equation*}
which is equivalent to
\begin{equation*}
  \forall\: e\subseteq\mathrm{supp}(P)\colon
  \big(2\leq |e|\leq L+1\implies e\notin\mathcal{E}\big).
\end{equation*} 

We next prove the upper bound on $I_{L+1}(G,P)$. Since each edge 
of $G$ has cardinality at least two, for all $v\in\mathcal{V}$, 
$\{v\}\notin\mathcal{E}$.
Thus
\begin{equation*}
  \sum_{v_{[L+1]}:\sigma(v_{[L+1]})\notin\mathcal{E}}P(v_{[L+1]})
  \geq \sum_{v\in\mathcal{V}} \big(P(v)\big)^{L+1},
\end{equation*}
which implies
\begin{equation*}
  I_{L+1}(G,P)\leq H_{L+1}(P),
\end{equation*}
where $H_{L+1}(P)$ is the R{\'e}nyi entropy of order $L+1$. 
Equality holds if and only if
\begin{equation*}
  \forall\: v_{[L+1]}\in\big(\mathrm{supp}(P)\big)^{L+1}\colon
  \big(v_1=\dots=v_{L+1}\big)
  \lor\big(\sigma(v_{[L+1]})\in\mathcal{E}\big),
\end{equation*}
which is equivalent to
\begin{equation*}
  \forall\: e\subseteq\mathrm{supp}(P)\colon
  \big(2\leq |e|\leq L+1\implies e\in\mathcal{E}\big).
\end{equation*} 

(iii) Fix a positive integer $n$. Let $\mathcal{E}_n$
denote the set of edges of $G^n$. Let $v_{[L+1]}^n$
denote the vector 
\begin{equation*}
  v_{[L+1]}^n:=\big(v_1^n,\dots,v_{L+1}^n\big),
\end{equation*}
and $\sigma_{L+1}(v_{[L+1]}^n)$ denote the set
\begin{equation*}
  \sigma_{L+1}(v_{[L+1]}^n):=
  \big\{v_1^n,\dots,v_{L+1}^n\big\}.
\end{equation*}
Furthermore, let $\mathcal{S}\subseteq\mathcal{V}^{L+1}$ denote the
set
\begin{equation*}
  \mathcal{S}:=\big\{v_{[L+1]}\big|
  \sigma_{L+1}(v_{[L+1]})\notin\mathcal{E}\big\}.
\end{equation*}
Note that for each $v_{[L+1]}^n$, 
$\sigma_{L+1}(v_{[L+1]}^n)\notin\mathcal{E}_n$
if and only if
\begin{equation*}
  \forall\: t\in [n]\colon
  \big\{v_{1t},\dots,v_{(L+1)t}\big\}\notin\mathcal{E}.
\end{equation*}
Thus
\begin{equation*}
  \big\{v_{[L+1]}^n\big|
  \sigma_{L+1}(v_{[L+1]}^n)\notin\mathcal{E}_n\big\}
  = \mathcal{S}^n,
\end{equation*}
which implies
\begin{align} \label{eq:multi2singleLetter}
  \sum_{v_{[L+1]}^n:\sigma_k(v_{[L+1]}^n)\notin\mathcal{E}_n}
  P^n(v_{[L+1]}^n)
  &= \sum_{v_{[L+1]}^n\in \mathcal{S}^n}P^n(v_{[k]}^n)\notag\\
  &= \sum_{v_{[L+1]}^n\in \mathcal{S}^n}\prod_{t\in [n]}P(v_{[k]t})\\
  &= \prod_{t\in [n]}\sum_{v_{[L+1]t}\in S}P(v_{[L+1]t})\notag\\
  &=\Big(\sum_{v_{[L+1]}\in \mathcal{S}}P(v_{[L+1]})\Big)^n,\notag
\end{align}
where in (\ref{eq:multi2singleLetter}), 
$v_{[L+1]t}=(v_{1t},\dots,v_{(L+1)t})$. 
Therefore, 
\begin{equation*}
  I_{L+1}(G^n,P^n)=nI_{L+1}(G,P).
\end{equation*}

For $\ell\in [L]$, we can write $\theta^{(\ell)}_{L+1}(G^n,P^n)$
as
\begin{equation*}
  \theta^{(j)}_{L+1}(G^n,P^n)
  =2I_{L+1}(G^n,P^n)+\frac{1}{L}\log
  \sum_{
  \substack{
  (v_{[\ell]}^n,v_{[\ell+1:L+1]}^n,\bar{v}_{[\ell+1:L+1]}^n):\\
  (v_{[\ell]}^n,v_{[\ell+1:L+1]}^n)\in \mathcal{S}^n\\
  (v_{[\ell]}^n,\bar{v}_{[\ell+1:L+1]}^n)\in \mathcal{S}^n\\
  }
  }P^n\big(v_{[\ell]}^n\big)P^n\big(v_{[\ell+1:L+1]}^n\big)
  P^n\big(\bar{v}_{[\ell+1:L+1]}^n\big).
\end{equation*}
Using a similar argument as above, it follows that for all
$\ell\in [L]$,
\begin{equation*}
  \theta^{(\ell)}_{L+1}(G^n,P^n)
  = n\theta^{(\ell)}_{L+1}(G,P).
\end{equation*}

\subsection{Proof of Theorem \ref{thm:main}} \label{subsec:thmMain}

We start by finding upper and lower bounds on the 
probability
that a random mapping from $[M]$ to $\mathcal{X}$ is 
an $(M,L)$ 
zero-error list code for the channel $W$.\footnote{Without loss 
of generality, we may assume that $M>L$, since if $M\leq L$,
then every mapping $f\colon [M]\rightarrow\mathcal{X}$ is
an $(M,L)$ zero-error list code.} The  
theorem then follows from applying our bounds to the 
channel $W^n$ for every positive integer $n$. 

Consider the random mapping
$F\colon [M]\rightarrow \mathcal{X}$, where $(F(m))_{m\in [M]}$
is a collection of i.i.d.\ random variables and each 
$F(m)$ has distribution
\begin{equation*}
  \pr\big\{F(m)=x\big\}:=P(x).
\end{equation*}
For every $S\in\binom{[M]}{L+1}$, define the random 
variable $Z_S$ as 
\begin{equation*}
  Z_S:=\mathbf{1}\Big\{\{F(m)\}_{m\in S}
  \notin \mathcal{E}\Big\};
\end{equation*}
that is, $Z_S$ is the indicator of the event that 
$\{F(m)\}_{m\in S}$ is not an edge of the distinguishability 
hypergraph $G(W)$. 
Let\footnote{For results regarding the distribution of $Z$ in the classical
birthday problem scenario, we refer the reader to the work of Arratia,
Goldstein, and Gordon \cite{ArratiaEtAl,ArratiaEtAl2}. A
direct application of the bounds in \cite{ArratiaEtAl,ArratiaEtAl2}
to $\pr\{Z=0\}$ leads to weaker results than those we present here.}
\begin{equation*}
  Z:=\sum_{S\in\binom{[M]}{L+1}} Z_S.
\end{equation*}
Note that by Proposition \ref{prop:existence},
$F$ is an $(M,L)$ zero-error list code if and only if $Z=0$.
The rest of the proof consists of computing
a lower and an upper bound for $\pr\{Z=0\}$. 

\textbf{Lower Bound.} By Markov's inequality,
\begin{align}
  \MoveEqLeft
  \pr\big\{F\text{ is an }(M,L)\text{ zero-error list code}\big\}
  \notag\\
  &= \pr\{Z=0\}\notag\\
  &= 1-\pr\big\{Z\geq 1\big\}\notag\\
  &\geq 1-\mathbb{E}[Z].\label{eq:MarkovIneqLB}
\end{align}
For any $S\in\binom{[M]}{L+1}$, 
\begin{equation} \label{eq:expZS}
  \mathbb{E}[Z_S]= 
  \sum_{x_{[L+1]}:\sigma(x_{[L+1]})\notin \mathcal{E}} 
  P(x_{[L+1]})=2^{-LI_{L+1}(G,P)}.
\end{equation}
By linearity of expectation,
\begin{equation} \label{eq:EZ}
  \mathbb{E}[Z]= \binom{M}{L+1}2^{-LI_{L+1}(G,P)}.
\end{equation}
Combining (\ref{eq:MarkovIneqLB}) and (\ref{eq:EZ})
gives
\begin{align*}
  \pr\big\{Z=0\big\}
  &\geq 1-\binom{M}{L+1}2^{-LI_{L+1}(G,P)}\\
  &\geq 1-M^{L+1}2^{-LI_{L+1}(G,P)},
\end{align*}
where the last inequality follows from the fact that 
\begin{equation*}
  \binom{M}{L+1}
  \leq M^{L+1}. 
\end{equation*}

\textbf{Upper Bound.}  We apply the second moment method.
By the Cauchy-Schwarz inequality, 
\begin{equation*}
  \mathbb{E}[Z]=\mathbb{E}\big[Z\mathbf{1}_{\{Z\geq 1\}}\big]
	\leq \sqrt{\mathbb{E}[Z^2]\times\pr\{Z\geq 1\}},
\end{equation*}
thus 
\begin{equation*}
  \pr\{Z\geq 1\}\geq \frac{\big(\mathbb{E}[Z]\big)^2}{\mathbb{E}[Z^2]}
\end{equation*}
or
\begin{equation*}
  \pr\{Z=0\} \leq 1-\frac{\big(\mathbb{E}[Z]\big)^2}{\mathbb{E}[Z^2]}.
\end{equation*}
To evaluate the upper bound on $\pr\{Z=0\}$, 
we calculate $\mathbb{E}[Z^2]$.
We have
\begin{align} \label{eq:Zsquared}
  Z^2 &= \Bigg[\sum_{S\in\binom{[M]}{L+1}}Z_S\Bigg]^2\notag\\
  &=\sum_{S\in\binom{[M]}{L+1}} Z_S^2+
  \sum_{\substack{S,S'\in\binom{[M]}{L+1}\\
  S\neq S'}} Z_SZ_{S'} \notag\\
  &=\sum_{S\in\binom{[M]}{L+1}} Z_S+
  \sum_{\ell=0}^L\sum_{\substack{S,S'\in\binom{[M]}{L+1}\\
  |S\cap S'|=\ell}} Z_SZ_{S'}.
\end{align}
For all $\ell\in\{0,1,\dots,L\}$, fix sets $S_\ell,S'_\ell\in\binom{[M]}{L+1}$
such that $|S_\ell\cap S_\ell'|=\ell$. When $\ell\in [L]$, 
$(F(m))_{m\in S_\ell}$ and $(F(m))_{m\in S_\ell'}$ are 
independent given $(F(m))_{m\in S_\ell\cap S_\ell'}$. Thus
for $\ell\in [L]$,
\begin{align}
  \mathbb{E}[Z_{S_\ell}Z_{S'_\ell}]
  &=\sum_{x_{[\ell]}}P(x_{[\ell]})
  \Bigg[\sum_{\substack{x_{[\ell+1:L+1]}:\\ \sigma(x_{[L+1]})
  \notin \mathcal{E}}}
  P(x_{[\ell+1:L+1]})\Bigg]^2\notag\\
  &= 2^{L(\theta_{L+1}^{(\ell)}(G,P)-2I_{L+1}(G,P))},
  \label{eq:ZSl}
\end{align}
where in (\ref{eq:ZSl}), we use the definition of 
$\theta_{L+1}^{(\ell)}(G,P)$ given by
(\ref{eq:defTkGP}). 
When $\ell=0$, 
$Z_{S_\ell}$ and $Z_{S_{\ell}'}$ are independent.
Thus by (\ref{eq:expZS}),
\begin{equation} \label{eq:ZS0ZS0Prime}
  \mathbb{E}[Z_{S_0}Z_{S_0'}]
  = \big(\mathbb{E}[Z_{S_0}]\big)^2=2^{-2LI_{L+1}(G,P)}.
\end{equation}
Equations (\ref{eq:Zsquared}), (\ref{eq:ZSl}), and
(\ref{eq:ZS0ZS0Prime}) together imply 
\begin{align}
  \mathbb{E}[Z^2] &= \binom{M}{L+1}2^{-LI_{L+1}(G,P)}
  +\binom{M}{0,L+1,L+1}2^{-2LI_{L+1}(G,P)}
  \notag\\
  &\phantom{=}+\sum_{\ell=1}^L \binom{M}{\ell,L+1-\ell,L+1-\ell}
  2^{L(\theta_{L+1}^{(\ell)}(G,P)-2I_{L+1}(G,P))},\label{eq:EZ2}
\end{align}
where in (\ref{eq:EZ2}), for $\ell\in \{0,1,\dots,L\}$,
the quantity
\begin{align*}
  \binom{M}{\ell,L+1-\ell,L+1-\ell}
  :=\binom{M}{\ell}\binom{M-\ell}{L+1-\ell}
  \binom{M-L-1}{L+1-\ell},
\end{align*}
equals the number of pairs $(S,S')$, where $S,S'\in \binom{[M]}{L+1}$
and $|S\cap S'|=\ell$. Combining (\ref{eq:EZ}) with 
(\ref{eq:EZ2}) now gives
\begin{align} \label{eq:cvSquared}
  \frac{\mathbb{E}[Z^2]}{\big(\mathbb{E}[Z]\big)^2} 
  &=\binom{M}{L+1}^{-1}2^{LI_{L+1}(G,P)}
  +\binom{M}{L+1}^{-2}\binom{M}{0,L+1,L+1}\notag \\
  &\phantom{=}+\binom{M}{L+1}^{-2}
  \sum_{\ell=1}^L \binom{M}{\ell,L+1-\ell,L+1-\ell}
  2^{L\theta_{L+1}^{(\ell)}(G,P)}\notag \\
  &\leq \binom{M}{L+1}^{-1}2^{LI_{L+1}(G,P)}
  +1+\sum_{\ell=1}^L \binom{L+1}{\ell}^2\binom{M}{\ell}^{-1}
  2^{L\theta_{L+1}^{(\ell)}(G,P)}.
\end{align}
The inequality in (\ref{eq:cvSquared}) follows
from the fact that for each $\ell\in\{0,1,\dots,L\}$, 
\begin{align*}
  \binom{M}{L+1}^{-2}\binom{M}{\ell,L+1-\ell,L+1-\ell}
  &=\binom{L+1}{\ell}^2\binom{M}{\ell}^{-1}\times
  \frac{(M-L-1)!(M-L-1)!}{(M-\ell)!(M-2L-2+\ell)!}\\
  &= \binom{L+1}{\ell}^2\binom{M}{\ell}^{-1}\prod_{j=\ell}^L
  \Big(\frac{M-L-1+\ell-j}{M-j}\Big)\\
  &\leq \binom{L+1}{\ell}^2\binom{M}{\ell}^{-1}.
\end{align*}
This completes the proof of the upper bound. 

The asymptotic result, as stated in Theorem \ref{thm:main}, 
follows from applying, for every $\ell\in [L+1]$, the inequality
\begin{equation*}
  \binom{M}{\ell}\geq
  \Big(\frac{M}{\ell}\Big)^\ell. 
\end{equation*} 

\subsection{Proof of Corollary \ref{cor:equality}}
\label{subsec:equality}

Corollary \ref{cor:equality} follows from applying the 
next lemma to the hypergraph $G^n(W)$ for sufficiently
large $n$. 

\begin{lem} Let $G=(\mathcal{V},\mathcal{E})$ be a
hypergraph and let 
$P$ be a distribution on $\mathcal{V}$. 
If $G_P$ is complete multipartite, then
\begin{equation} \label{eq:maxEllTheta}
  \max_{\ell\in [L+1]}\frac{1}{\ell}
  \theta^{(\ell)}_{L+1}(G,P)
  =\frac{1}{L+1}I_{L+1}(G,P).
\end{equation}
\end{lem}
\begin{proof}
Since $G_P=(\mathcal{V}_P,\mathcal{E}_P)$ 
is complete multipartite, there exists
a partition of $\mathcal{V}_P$ consisting of 
independent sets. Let $\{\mathcal{I}_{j}\}_{j=1}^k$
denote such a partition. 
Define the distribution $P^*$ on $[k]$ as 
\begin{equation*}
  P^*(j):=\sum_{v\in\mathcal{I}_j}P(v).
\end{equation*}
In words, $P^*(j)$ is the weight assigned to the independent
set $\mathcal{I}_j$ by $P$.  
Since $G_P$ is a complete multipartite hypergraph, 
$\sigma(v_{[L+1]})\notin\mathcal{E}$ if and only if 
there exists some $j$ such that 
$\sigma(v_{[L+1]})\subseteq\mathcal{I}_j$. Thus
\begin{align*}
  I_{L+1}(G_P,P) &=
  -\frac{1}{L}\log\sum_{v_{[L+1]}:\sigma(v_{[L+1]})\notin\mathcal{E}}
  P(v_{[L+1]})\\
  &= -\frac{1}{L}\log\sum_{j=1}^k\sum_{\sigma(v_{[L+1]})\subseteq\mathcal{I}_j}
  P(v_{[L+1]})\\
  &= -\frac{1}{L}\log\sum_{j=1}^k \big(P^*(j)\big)^{L+1}\\
  &= H_{L+1}(P^*),
\end{align*}
where $H_{L+1}$ denotes the R\'{e}nyi entropy of order $L+1$.
Similarly, for all $\ell\in [L+1]$ we have
\begin{equation} \label{eq:thetaMP}
  \theta_{L+1}^{(\ell)}(G_P,P)=2H_{L+1}(P^*)-
  \frac{2L+1-\ell}{L}H_{2L+2-\ell}(P^*).
\end{equation}

Using (\ref{eq:thetaMP}), we see that proving (\ref{eq:maxEllTheta}) 
is equivalent to showing that for all $\ell\in [L+1]$,
\begin{equation*}  
  \Big(\sum_{j=1}^k P^*(j)^{2L+2-\ell}\Big)
  ^\frac{1}{2L+2-\ell}\leq
  \Big(\sum_{j=1}^k P^*(j)^{L+1}\Big)
  ^\frac{1}{L+1},
\end{equation*}
which follows from the well-known fact that for all $p\geq q$, 
the $q$-norm dominates the $p$-norm.

Finally, note that for all $\ell\in [L+1]$, 
\begin{equation*}
  \theta_{L+1}^{(\ell)}(G_P,P)
  =\theta_{L+1}^{(\ell)}(G,P).
\end{equation*}
This completes the proof. 
\end{proof}

\subsection{Proof of Lemma \ref{lem:completeMultipartite}}
Since $G$ is complete multipartite,
there exists a partition of its set of vertices,
say $\{\mathcal{I}_{j}\}_{j=1}^k$, consisting of 
independent sets. 
Next note that for any $n\geq 2$, the set of vertices 
of $G^n$ is given by 
\begin{equation*}
  \bigcup_{j_1,\dots,j_n\in [k]}
  \mathcal{I}_{j_1}\times\dots\times
  \mathcal{I}_{j_n}
\end{equation*}
We show that an arbitrary subset of 
$\mathcal{V}^n$, say $\{v_1^n,\dots,v_\ell^n\}$, 
is an edge in $G^n$ if and only if 
\begin{equation} \label{eq:nDimIndep}
  \forall\: j_1,\dots,j_n\in [k]\colon
  \{v_1^n,\dots,v_\ell^n\}
  \not\subseteq \mathcal{I}_{j_1}\times\dots
  \times\mathcal{I}_{j_n}.
\end{equation}
By definition, $\{v_1^n,\dots,v_\ell^n\}$
is an edge in $G^n$ if and only if for some
$t\in [n]$, $\{v_{1t},\dots,v_{\ell t}\}$ is an
edge in $G$. Since $G$ is complete multipartite,
the latter condition holds if and
only if
\begin{equation*}
  \exists\: i\in [n]\text{ such that }
  \forall\:j\in [k]\colon
  \{v_{1i},\dots,v_{\ell i}\}
  \not\subseteq \mathcal{I}_j,
\end{equation*}
which is equivalent to (\ref{eq:nDimIndep}).

\section{Conclusion}

From Shannon's random coding argument \cite{Shannon48} 
it follows that if the rate of a randomly generated 
codebook is less than the input-output mutual information, 
the probability that the codebook has small probability 
of error goes to one as the blocklength goes to infinity. 
In this work, we find necessary and sufficient conditions on 
the rate so that the probability that a randomly generated
codebook has \emph{zero} probability of error goes to 
one as the blocklength goes to infinity. We further show 
that this result extends the classical
birthday problem to an information-theoretic setting and 
provides an intuitive meaning for R{\'e}nyi entropy.

\section*{Appendix A\\Properties of $I_2(G,P)$}
\label{app:Motzkin Straus}

In this appendix, we describe two properties of $I_2(G,P)$. 
In the first part, we state the Motzkin-Straus theorem \cite{MotzkinStraus}, which
gives the maximum of $I_2(G,P)$ over all distributions $P$
for a fixed graph $G$. In the second part, we show that
$I_2(G,P)$ is always less than or equal to 
K\"{o}rner's graph entropy \cite{Korner}.

\subsection{The Motzkin-Straus Theorem}
Consider a graph $G=(\mathcal{V},\mathcal{E})$. 
Motzkin and Straus \cite{MotzkinStraus}
prove that 
\begin{equation*}
  \max_P I_2(G,P)=\log\omega (G),
\end{equation*}
where the maximum is over all distributions $P$ defined on
$\mathcal{V}$, and $\omega (G)$ is the cardinality of the 
largest clique in $G$. An implication of this 
result is Tur{\'a}n's graph theorem \cite{Aigner}, 
which states that 
\begin{equation} \label{eq:Turan}
  \omega (G)\geq \frac{|\mathcal{V}|^2}
  {|\mathcal{V}|^2-2|\mathcal{E}|}.
\end{equation}
To see this, let $P$ be the uniform distribution on $\mathcal{V}$. 
Then
\begin{equation*}
  I_2(G,P)=-\log\sum_{v,v':\{v,v'\}\notin\mathcal{E}}
  \frac{1}{|\mathcal{V}|^2}=
  \log\frac{|\mathcal{V}|^2}{|\mathcal{V}|^2-2|\mathcal{E}|},
\end{equation*}
and (\ref{eq:Turan}) follows by the Motzkin-Straus theorem.
We remark that extensions of the Motzkin-Straus theorem
to hypergraphs are presented in \cite{SosStraus, FranklRodl, RotaPelillo}.

From Proposition \ref{prop:thetaProperties},
Part (ii) it follows that for any 
distribution $P$ defined on a set
$\mathcal{V}$, 
\begin{equation} \label{eq:maxOverGraphs}
  \max_{G}I_2(G,P)=H_2(P),
\end{equation}
where the maximum is over all graphs $G=(\mathcal{V},\mathcal{E})$.
In (\ref{eq:maxOverGraphs}), the maximum
is achieved when $G$ is the complete graph on the support of $P$.

\subsection{Relation with K\"{o}rner's Graph Entropy}
Consider a graph $G$ with vertex set $\mathcal{X}$.
Let $P$ be a probability distribution on $\mathcal{X}$
and $\mathcal{Y}\subseteq 2^\mathcal{X}$ be the 
set of all maximal independent subgraphs
of $G$. Let $\Delta(G,P)$ denote 
the set of all probability distributions $P(x,y)$ on 
$\mathcal{X}\times\mathcal{Y}$ whose marginal on $\mathcal{X}$
equals $P(x)$, and 
\begin{equation*}
  \pr\big\{X\in Y\big\}=\sum_{(x,y):x\in y}P(x,y)=1.
\end{equation*}
For the graph $G$ and probability distribution
$P$,  K\"{o}rner's graph entropy \cite{Korner} is defined 
by 
\begin{equation} \label{eq:KornerEntropy}
  H_1(G,P)=\min_{\Delta(G,P)} I(X;Y).
\end{equation}
Our aim is to define $H_2(G,P)$ in a similar manner to 
how $H_2(P)$, the  R\'{e}nyi entropy of order 2, is defined. 
One way to accomplish 
this task is through the use of Jensen's inequality. 
Applying Jensen's inequality to Shannon entropy gives
\begin{align*}
  H_1(P) &= -\sum_{x}P(x)\log P(x)\\
  &\geq -\log\Big(\sum_{x}\big(P(x)\big)^2\Big)=H_2(P).
\end{align*}
Analogously, applying Jensen's inequality to the 
mutual information in (\ref{eq:KornerEntropy}) gives
\begin{align*}
  I(X;Y) &= -\sum_{(x,y):x\in y}P(x,y)\log \frac{P(x)P(y)}{P(x,y)}\\
  &\geq -\log\Big(\sum_{(x,y):x\in y} P(x)P(y)\Big),
\end{align*}
Thus we define $H_2(G,P)$ as 
\begin{align*}
  H_2(G,P)&=\min_{\Delta(G,P)} 
  -\log\Big(\sum_{(x,y):x\in y} P(x)P(y)\Big)\\
  &\leq H_1(G,P).
\end{align*}
Our next proposition relates $H_2(G,P)$ and $I_2(G,P)$.
\begin{prop} \label{prop:KornerEntropy}
For any graph $G$ and any probability distribution $P$ 
defined on its vertices, $I_2(G,P)\leq H_2(G,P)$. 
\end{prop}
\begin{proof}
Let $P(x,y)\in \Delta(G,P)$. Then
\begin{equation} \label{eq:twoReps}
  \sum_{x\in y} P(x)P(y)
  = \sum_{x,x'}P(x)P(x')\sum_{y}P(y|x')
  \mathbf{1}\big\{x,x'\in y\big\}.
\end{equation}
Since every $y$ is an independent subgraph of $G$,
if $x,x'\in y$, then $(x,x')\notin \mathcal{E}$. 
Thus 
\begin{equation*}
  \mathbf{1}\big\{x,x'\in y\big\}
  \leq \mathbf{1}\big\{(x,x')\notin \mathcal{E}\big\},
\end{equation*}
which implies
\begin{equation*}
  \sum_y P(y|x')\mathbf{1}\big\{x,x'\in y\big\}
  \leq \mathbf{1}\big\{(x,x')\notin \mathcal{E}\big\}.
\end{equation*}
By (\ref{eq:twoReps}), we have
\begin{equation*}
  \sum_{x\in y}P(x)P(y)
  \leq \sum_{x,x'}P(x)P(x')\mathbf{1}
  \big\{(x,x')\notin \mathcal{E}\big\}.
\end{equation*}
Calculating the logarithm of both sides and maximizing the left
hand side over $\Delta(G,P)$ gives
$H_2(G,P)\geq I_2(G,P)$. 
\end{proof}

\section*{Appendix B\\Connection to R{\'e}nyi's Result}

We next state R{\'e}nyi's result \cite[Equation (5.3)]{Renyi2}
in the context of zero-error list codes over the identity channel.
For the identity channel setting described above, 
fix $\epsilon\in (0,1)$, positive integer $L\geq 1$, and 
probability mass function $P$ on $\mathcal{X}$. 
For every positive integer $n$ and $x^n\in\mathcal{X}^n$, 
define
\begin{equation} \label{eq:iidDistP}  
  P_n(x^n)\coloneqq \prod_{i=1}^n P(x_i).
\end{equation}
For positive integers $M$ and $n$, let $\Phi_{M,n}\colon [M]\rightarrow\mathcal{X}^n$
be a random mapping with i.i.d.\ values $\Phi_{M,n}(1),\dots,\Phi_{M,n}(M)$, each
with distribution $P_n(x^n)$. Define $n^*_{L+1}(M,\epsilon)$
as the least positive integer $n$ for which
\begin{equation*}
  \pr\big\{\Phi_{M,n}\text{ is an }(M,n,L)\text{ zero-error list code}
  \big\}\geq 1-\epsilon.
\end{equation*}
In \cite{Renyi2}, R{\'e}nyi states that
\begin{equation} \label{eq:RenyiBDay}
  \lim_{M\rightarrow\infty}
  \frac{\log M}{n^*_{L+1}(M,\epsilon)}
  =\frac{L}{L+1}\cdot H_{L+1}(P).
\end{equation}

We next describe a connection between R{\'e}nyi's
result and our result concerning the birthday problem. Consider
the scenario where the distribution of each codeword 
is given by (\ref{eq:iidDistP}), the size of the message
set equals $M_n\coloneqq\lfloor 2^{nR}\rfloor$ for some
$R\geq 0$, and (\ref{eq:bDayLimit})
holds. We show that both (\ref{eq:OurResultBDay})
and (\ref{eq:RenyiBDay}) imply\footnote{More precisely,
(\ref{eq:OurResultBDay}) gives (\ref{eq:BDayImplication}) 
with strict inequality.}
\begin{equation} \label{eq:BDayImplication}
  R\leq\frac{L}{L+1}\cdot H_{L+1}(P).
\end{equation}

First note that for $P_n$ given by (\ref{eq:iidDistP}),
(\ref{eq:OurResultBDay}), 
together with Part (iii) of 
Proposition \ref{prop:thetaProperties},
gives
\begin{equation*}
  \lim_{n\rightarrow\infty}
  \Big[(L+1)\log \lfloor 2^{nR}\rfloor -nLH_{L+1}(P)\Big]=-\infty,
\end{equation*}
which directly leads to (\ref{eq:BDayImplication}).

Next we show that
(\ref{eq:RenyiBDay}) implies (\ref{eq:BDayImplication}).
First note that for all $n$,
$F_n$ has the same distribution as $\Phi_{M_n,n}$.
In addition, by assumption, for sufficiently large $n$, 
\begin{equation*}
  \pr\big\{F_n\text{ is an }(M_n,n,L)\text{ zero-error list code}
  \big\}\geq 1-\epsilon.
\end{equation*}
Therefore, by the definition of $n^*_{L+1}$,
we have
\begin{equation*}
  n\geq n^*_{L+1}(M_n,\epsilon).
\end{equation*}
Thus if $R>0$, 
\begin{equation*}
  \lim_{M\rightarrow\infty}
  \frac{\log M}{n^*_{L+1}(M,\epsilon)}
  = \lim_{n\rightarrow\infty}
  \frac{\log \lfloor 2^{nR}\rfloor}
  {n^*_{L+1}(\lfloor 2^{nR}\rfloor,\epsilon)}
  \geq \lim_{n\rightarrow\infty}
  \frac{1}{n}\log \lfloor 2^{nR}\rfloor
  =R.
\end{equation*}
Applying (\ref{eq:RenyiBDay}) completes the proof. 

\section*{Acknowledgments}
The first author thanks Ming Fai Wong for helpful 
discussions regarding an earlier version of Theorem \ref{thm:main}.

\bibliographystyle{IEEEtran}
\bibliography{ref}{}

\end{document}